\documentclass[submission,copyright,creativecommons]{eptcs}
\providecommand{\event}{IWIGP 2011} 
\usepackage{breakurl}             
\usepackage{amsmath}
\usepackage{amsfonts}
\usepackage{amssymb}
\usepackage{graphicx}
\usepackage{graphicx}
\usepackage{marvosym}
\usepackage{textcomp}
\newtheorem{theorem}{Theorem}

\newtheorem{proposition}{Proposition}
\newtheorem{corollary}{Corollary}
\newtheorem{definition}{Definition}

\newenvironment{proof}[1][Proof]{\begin{trivlist}
\item[\hskip \labelsep {\bfseries #1}]}{\end{trivlist}}
\newenvironment{example}[1][Example]{\begin{trivlist}
\item[\hskip \labelsep {\bfseries #1}]}{\end{trivlist}}

\def\fract#1/#2{\leavevmode
 \kern.1em \raise .5ex \hbox{\the\scriptfont0 #1}%
 \kern-.1em $/$%
 \kern-.15em \lower .25ex \hbox{\the\scriptfont0 #2}%
}
\def\abs#1{\ensuremath{\lvert #1\rvert}}
\def\norm#1{\ensuremath{\lVert #1\rVert}}
\newcommand{\nat}{\mathbb N}
\newcommand{\tuple}[1]{\langle #1 \rangle}
\newcommand{\dist}{{\cal D}}
\newcommand{\Supp}{{\sf Supp}}
\newcommand{\powset}{{\cal{P}}}
\newcommand{\Plays}{{\sf Plays}}
\newcommand{\Hists}{{\sf Hists}}
\newcommand{\Last}{{\sf Last}}
\newcommand{\Inf}{{\sf Inf}}
\renewcommand{\P}{{\sf P}}
\newcommand{\B}{{\sf B}}
\newcommand{\Action}{{\sf Action}}
\newcommand{\State}{{\sf State}}
\newcommand{\q}{{\sf q}}

\title{Synchronizing Objectives for Markov Decision Processes}
\author{Laurent Doyen
\institute{LSV, ENS Cachan \& CNRS, France}
\email{doyen@lsv.ens-cachan.fr} \and Thierry Massart \qquad\qquad
Mahsa Shirmohammadi \institute{Universit\'e Libre de Bruxelles,
Brussels, Belgium\thanks{This work has been done in the MoVES
project (P6/39) which is part of the IAP-Phase VI Interuniversity
Attraction Poles Programme funded by the Belgian State, Belgian
Science Policy.}} \email{\quad thierry.massart@ulb.ac.be
\quad\qquad mahsa.shirmohammadi@ulb.ac.be} }
\def\titlerunning{Synchronizing objectives for Markov Decision Processes}
\def\authorrunning{L. Doyen, T. Massart \& M. Shirmohammadi}
\begin{document}
\maketitle
\begin{abstract}
{\bf Abstract.} We introduce synchronizing objectives for Markov
decision processes (MDP). Intuitively, a synchronizing objective
requires that eventually, at every step there is a state 
which concentrates almost all the probability mass. 
In particular, it implies that the probabilistic system behaves in
the long run like a deterministic system: eventually, the current state of the MDP
can be identified with almost certainty.

We study the problem of deciding the existence of a strategy to
enforce a synchronizing objective in MDPs. We show that the
problem is decidable for general strategies, as well as for blind
strategies where the player cannot observe the current state of
the MDP. We also show that pure strategies are sufficient, but
memory may be necessary.
\end{abstract}
\section{Introduction}
A \emph{Markov decision process (MDP)} is a model for systems that
exhibit both probabilistic and nondeterministic behavior. MDPs
have been used to model and solve control problems for stochastic
systems where the nondeterminism represents the freedom of the
controller to choose a control action, while the probabilistic
component of the behavior describes the system response to control
actions. MDPs have also been adopted as models for concurrent
probabilistic systems, probabilistic systems operating in open
environments~\cite{SegalaT}, and under-specified probabilistic
systems \cite{BiancoA95}.


 Traditional objectives for MDP specify a set
$S$ of paths, where a path is an infinite sequence of states
through the underlying graph of the MDP. The value of interest is
the probability that an execution of the MDP under a given
strategy belongs to $S$. For example, a reachability objective
specifies all paths that visit a given target state $\ell$. A
typical qualitative question is to decide whether there exists a
strategy such that  a given state $\ell$ is reached with probability
$1$.


In this paper, we consider a different type of objectives which
specify a set of infinite sequences $\bar{X} = X_0, X_1, \dots$ of
probability distributions over the states~\cite{KVKA10}.
Intuitively, the distribution $X_i$ in the sequence gives for each
state $\ell$ the probability $X_i(\ell)$ to be in state $\ell$ at
step $i \geq 0$. We introduce \emph{synchronizing objectives}
which specify sequences of distributions in which the probability
tends to accumulate in a single state. We use the infinity norm as
a measure of the highest peak in a probability distribution~$X_i$
(i.e., $\norm{X_i} = \max_{\ell \in L} X_i(\ell)$) and we require
that the limit\footnote{ Since the limit may not exist in general,
we actually consider either $\liminf$ or $\limsup$.} of this
measure in the sequence is $1$. Intuitively, this requires that in
the long run, the MDP behaves like a deterministic system:
from some point on, at every step $i$ there is a state $\ell_i$ 
which accumulates almost all the probability. 
Note that satisfying such an
objective implies that there exists a state $\ell$ which is
reached with probability $1$. The converse does not hold because
reachability objectives do not require the visits to the target
state to occur after the same number of steps in (almost) all
executions of the MDP. We consider the problem of deciding if a
given MDP is synchronizing
for some strategy, 
We consider the general case where memoryful randomized strategies
are allowed, as well as the special case of blind strategies which
are not allowed to observe the current state of the MDP.


Defining objectives as a sequence of probability distributions
over states rather than a distribution over sequences of states is
a change of standpoint in the traditional approach to MDP
verification. Up to our knowledge, there are very few works in
this setting. We are aware of the work in~\cite{KVKA10} which
studies MDPs as generators of probability distributions with
applications in sensor networks and dynamical systems, and shows
that the resulting objectives are not expressible in known logics
such as PCTL${}^*$~\cite{AzizSB95,BiancoA95}. In their definition,
probability distributions over states are assigned a vector $v \in
\{0,1\}^k$ of truth values for a finite set of predicates
$\varphi_1, \dots, \varphi_k$ (which are linear constraints on the
probabilities such as $\varphi(X) \equiv X(\ell) + X(\ell') \leq
\frac{1}{2}$, for example). This can be viewed as a coloring of
the probability distributions using a finite number of colors, and
then objectives are languages of infinite words over the finite
alphabet of colors. It is shown that reachability of a given color
is undecidable for MDPs if arbitrary linear predicates are
allowed~\cite{KVKA10}. A decidability result is obtained if only
predicates of the form $\sum_{\ell \in T} X(\ell) > 0$ are
allowed. Synchronizing objectives cannot be expressed in the
framework of~\cite{KVKA10} using finite colorings as they require
a real-valued measure (namely, the infinite norm) to be assigned
to the probability distributions.


In~\cite{BeauquierRS02}, the monadic logic of probabilities is
introduced as a predicate logic which can express properties of
sequences of probability distributions. But because it allows
comparison of probabilities only with constants, it cannot express
synchronizing objectives which would require a quantification over
probability thresholds, such as
$\varphi(\bar{X}) \equiv \forall \epsilon > 0 \cdot \exists N
\cdot \forall i \geq N \cdot \exists \ell \in L: X_i(\ell) \geq 1-
\epsilon$, where $X_i$ is the probability distribution in position
$i$ in the sequence $\bar{X}$.



Synchronizing objectives generalize the notion of synchronizing
words. In a deterministic finite automaton, a word $w$ is
synchronizing if reading $w$ from any state of the automaton
always leads to the same state. It is sufficient to consider
finite words, and it is conjectured that if a synchronizing word
exists, then there exists one of length at most $(n-1)^2$ where
$n$ is the number of states of the automaton, known as the
\v{C}ern\'{y}'s conjecture. Several works have studied this
conjecture and related problems (see the survey
in~\cite{Volkov08}). Viewing deterministic automata as a special
case of MDP where all transitions have only one successor, a
synchronizing word can be seen as a blind strategy to ensure a
synchronizing objective. Note that we do not present a
generalization of \v{C}ern\'{y}'s conjecture since in our case,
strategies for MDPs are infinite objects. However, synchronizing
objectives provide an extension of the design framework for the
many applications of the theory of synchronizing words, such as
control of discrete event systems, planning, biocomputing, and
robotics~\cite{Volkov08}. For example, in probabilistic models of
DNA transcription, one may ask which molecules to introduce in a
cell in order to bring it to a single possible
state~\cite{Ben03,Volkov08}.


We prove that it is decidable to determine if a given MDP is
synchronizing for some strategy, either blind or general. We use
variants of the subset construction in the underlying graph of
MDPs to obtain a decidable characterization of synchronizing
strategies. Our results imply that pure strategies are sufficient
to satisfy a synchronizing objective, but we provide an example
showing that memory may be necessary, both with blind and general
strategies.


\section{Definitions.}
A \emph{probability distribution} over a finite set~$S$ is a
function $d : S \rightarrow [0, 1]$ such that $\sum_{s \in S} d(s)
= 1$. The \emph{support} of~$d$ is the set $\Supp(d) = \{s \in S
\mid d(s) > 0\}$.
 $\dist(S)$ denotes the set of all probability distributions
 on~$S$, and
$\powset(S)$ the power set of~$S$.


\paragraph{Markov decision processes.}
A \emph{Markov decision process} (MDP) is a tuple
$M=\tuple{L,\mu_{0},\Sigma,\delta}$ where $L$ is a finite set of
states, $\mu_0 \in \dist(L)$ is an initial probability
distribution over states, $\Sigma$ is a finite set of actions,
$\delta: L \times \Sigma \to \dist(L)$ is a probabilistic transition
function that assigns to each pair of states and actions, a
probability distribution over successor states.
A \emph{Markov chain} 
is a special case of MDPs with only one action ($\abs{\Sigma}=1$).
Markov chains are therefore generally viewed as a tuple $M =
\tuple{L,\mu_0,\delta}$ where $\delta: L \to \dist(L)$.
For an action~$\sigma \in \Sigma$ and a state~$\ell \in L$,
let $Post_{\sigma}(\ell)=\Supp(\delta(\ell,\sigma))$, 
and for a set~$s \subseteq L$, let $Post_{\sigma}(s)=\cup _{\ell \in s} \, Post_{\sigma}(\ell)$.



\begin{example}
\figurename~\ref{f5}(a) shows an MDP with four states and alphabet
$\Sigma = \{\sigma_1,\sigma_2\}$. The initial probability
distribution is $\mu_0(1)=1$ and $\mu_0(i)=0$ for $i \in
\{2,3,4\}$, and the probabilistic transition function $\delta$ in
state~$1$ is such that $\delta(1,\sigma_1)(2) =
\delta(1,\sigma_1)(3) = 1/2$ and $\delta(1,\sigma_2)(1)=1$.
\end{example}


We describe the behavior of an MDP as a one-player stochastic game
played for infinitely many rounds. In the first round, the game
starts in state $\ell$ with probability $\mu_{0}(\ell)$. In each
round, if the game is in the state $\ell$ and the player chooses
the action $\sigma \in \Sigma$, then the game moves to the
successor state $\ell'$ chosen with probability
$\delta(\ell,\sigma)(\ell')$, and the next round starts.
We consider two versions of this game. In both versions, the player knows the 
structure of the MDP. In the first version the
player has \emph{perfect information}, he can see the current
state of the game; in the second version the player is \emph{blind}, he is
not allowed to observe the current state of the game, and only knows 
the number of rounds that have been played so far.



A \emph{play} of the game is an infinite sequence of interleaved
states and actions $\pi=\ell_{0}\sigma_{0}\ell_{1} \cdots$ such
that $\ell_{i+1} \in Post_{\sigma_{i}}(\ell_{i})$ for all $i \geq 0$. 
The set of all plays over $M$ is denoted by~$\Plays(M)$. A
finite prefix $h=\ell_{0}\sigma_{0}\ell_{1} \cdots
\sigma_{n-1}\ell_{n}$ of a play~$\pi$ is called a \emph{history},
the last state of~$h$ is $\Last(h)=\ell_{n}$, the i$^{th}$ action and 
state of the of~$h$ is $\Action(h,i)=\sigma_{i}$ and $\State(h,i)=\ell_{i}$, 
and its length is $\abs{h} = n$. The set of all histories of plays is 
denoted by $\Hists(M)$.


\paragraph{Strategies and outcome.}
In the game, the choice of the action is made by the player
according to a strategy. Depending on what the player can observe
and record, he can use various classes of strategies. A
\emph{randomized strategy} (or simply a strategy) over an MDP~$M$
is a function $\alpha: \Hists(M) \to \dist(\Sigma)$. A \emph{pure}
(deterministic) strategy is a special case of randomized strategy
where for all $h \in \Hists(M)$, there exists an action $\sigma
\in \Sigma$ such that $\alpha (h)(\sigma)=1$.
A \emph{memoryless}
strategy is a randomized strategy $\alpha$ such that $\alpha(h_{1})=\alpha(h_{2})$
 for all $h_{1}, h_{2} \in
\Hists(M)$ with $\Last(h_{1})=\Last(h_{2})$. In this last case,  the player cannot
record the history of the play and makes a choice according to the
current state only. For convenience, we view pure strategies as
functions $\alpha : \Hists(M) \to \Sigma$, and memoryless
strategies as functions $\alpha: L \to \dist(\Sigma)$. Hence, a
pure memoryless strategy is a function $\alpha: L \to \Sigma$.


A strategy $\alpha$ is \emph{blind} if $\alpha(h_1) = \alpha(h_2)$
for all $h_{1}, h_{2} \in \Hists(M)$ such that $\abs{h_1} =
\abs{h_2}$. Blind strategies can be viewed as functions $\alpha:
\nat \to \dist(\Sigma)$ (or, $\alpha: \nat \to \Sigma$ for pure
blind strategies) which assign in each round a probability
distribution over actions. Sometimes we talk about
\emph{perfect-information} strategies to emphasize when we
consider strategies that are not necessarily blind.


The \textit{outcome} of the game played on an MDP $M =
\tuple{L,\mu_{0},\Sigma,\delta}$ using a strategy $\alpha$ is the
infinite sequence $X^{\alpha}_{0} X^{\alpha}_{1} \dots $ of
probability distributions over the set of states $L$, where
$X^{\alpha}_{0} = \mu_{0}$ and for all $n > 0$,
$$ \textstyle X_{n}^{\alpha}(\ell)=\sum_{h \in \Hists(M): \Last(h)=\ell,\abs{h}=n}Pr^{\alpha}(h)$$
where the probability $Pr^{\alpha}(h)$ of a history $h=\ell_{0}\sigma_{0}\ell_{1}
\cdots \sigma_{n-1}\ell_{n}$ under strategy $\alpha$ is
$$ \textstyle Pr^{\alpha}(h)= \mu_0( \ell_{0}) \cdot \prod_{j=1}^{n}
\alpha(\ell_{0} \sigma_{0}\dots \ell_{j-1})(\sigma_{j-1}) \cdot
\delta(\ell_{j-1},\sigma_{j-1})(\ell_{j}).$$


\paragraph{Synchronizing objectives.}
The \emph{norm} of a probability distribution~$X$ over~$L$ is
$\norm{X} = \max_{\ell \in L} X(\ell)$. We say that the MDP~$M$
with strategy~$\alpha$ is \emph{strongly synchronizing} if
\begin{equation}
 \liminf_{n \to \infty} \norm{X^{\alpha}_{n}} = 1,
\end{equation}
and that it is \emph{weakly synchronizing} if
\begin{equation}
 \limsup_{n \to \infty} \norm{X^{\alpha}_{n}} = 1.
\end{equation}

Intuitively, an MDP is synchronizing if the probability mass tends
to concentrate in a single state, either at every step from some point on (for strongly
synchronizing), or at infinitely many steps (for weakly synchronizing). 
Note that equivalently, $M$ with strategy $\alpha$ is strongly synchronizing if the limit
$\lim_{n \to \infty} \norm{X^{\alpha}_{n}}$ exists and equals~$1$.
In this paper, we are interested in the problem of deciding if a
given MDP is synchronizing for some strategy. We consider the
problem for both perfect-information and blind strategies.


\paragraph{Recurrent and transient states.}
A state $\ell' \in L$ is \emph{accessible} from a state $\ell \in
L$ (denoted $\ell \rightarrow \ell'$), if there is a history
$h=\ell_{0}\sigma_{0}\ell_{1} \cdots \sigma_{n-1}\ell_{n}$ with
$\ell_{0} = \ell$ and $\ell_{n} = \ell'$.
If both $\ell \rightarrow \ell'$ and $\ell' \rightarrow \ell$
hold, then we say that $\ell$ and $\ell'$ are strongly connected
(denoted $\ell \leftrightarrow \ell'$). This induces an
equivalence relation called \emph{accessibility relation}. An MDP
is \emph{strongly connected}, if all pairs of states $\ell, \ell'
\in L$ are strongly connected. A state accessible from a state of $\Supp(\mu_0)$ is simply called accessible state.


For a Markov chain $M$, the state $\ell$ is \emph{recurrent} if
all accessible states from $\ell$ can access $\ell$ (i.e., $\ell$
and $\ell'$ are strongly connected for all $\ell'$ such that $\ell
\rightarrow \ell'$), and the state $\ell$ is \emph{transient} if
there exists some state $\ell'$ such that $\ell'$ is accessible
from $\ell$, but $\ell$ is not accessible from $\ell'$.
The next proposition follows from standard results~\cite{FV97}.

\begin{proposition}\label{proposition-mc}
Given a Markov chain~$M$, let $X_0, X_1, \dots$ be the sequence of
probability distributions of~$M$. Then $\limsup_{n \to \infty}
X^n (\ell) = 0$ for all transient states $\ell \in L$, and
$\limsup_{n \to \infty} X^n (\ell) > 0$ for all recurrent states
$\ell \in L$.
\end{proposition}

\paragraph{Subset constructions.}
We define two important constructions based on the subset
construction idea. Subset construction is a standard technique to
compute, from a nondeterministic finite automaton~$N$, an
equivalent deterministic automaton~$D$ (for language equivalence),
where one state of $D$ corresponds to the set of possible states
(called a cell) in which $N$ can be. We define two kinds of subset
constructions on MDPs, the \emph{perfect-information subset
construction}, and the \emph{blind subset construction}. As usual,
each state of the subset constructions is a subset of states of
the MDP (i.e., a cell). In our case, the main difference lies in
the alphabet. In the perfect-information subset construction,
the selection of the next action depends on the current state (each
state of a cell can independently choose an action), while in 
the blind subset constructions the next action is independent 
of the state (all states of a cell have to choose the same action).
Thus, an action in the perfect-information subset
construction is a function $\hat{\sigma}: L \to \Sigma$ which
assigns to each state $\ell \in L$ its choice among the actions in
$\Sigma$.

\begin{figure}[t]
\begin{center}
\includegraphics[scale=0.5]{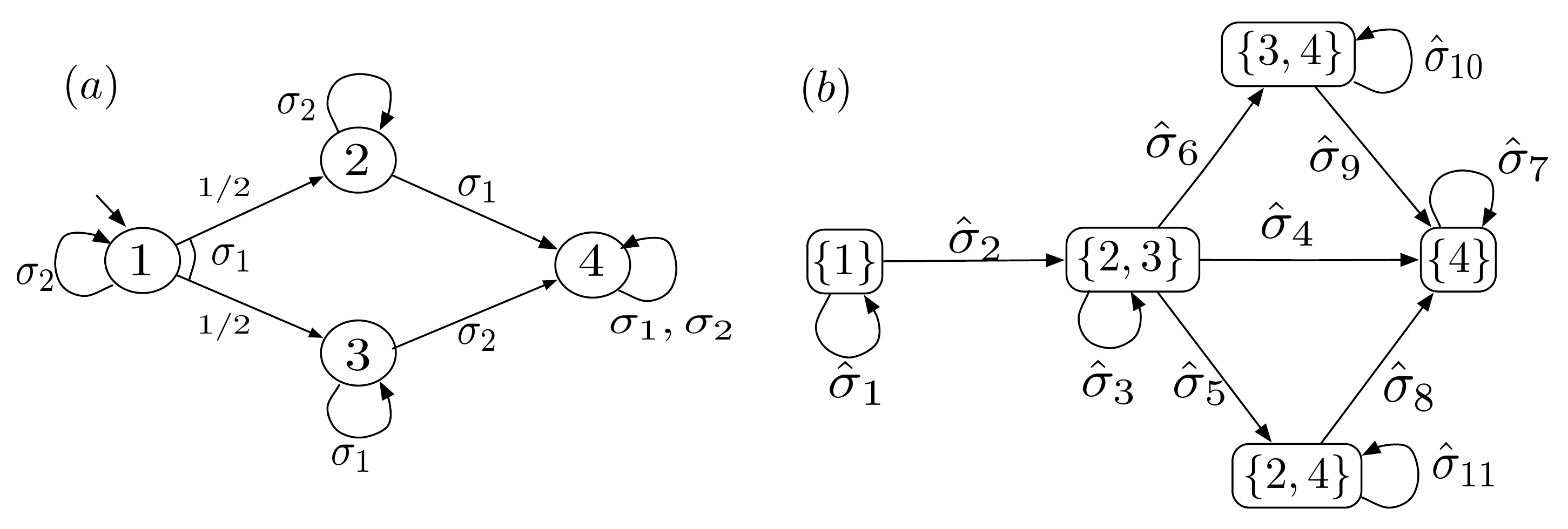}
\caption{(a) shows an MDP, and (b) shows the accessible states of its perfect information subset
construction.}\label{f5}
\end{center}
\end{figure}

\begin{definition}[Perfect-information subset construction of an MDP]
For an  MDP $M=~\tuple{L,\mu_{0},\Sigma,\delta},~$ the perfect-information subset construction is an automaton
 $M^{\P}=\tuple{\mathcal{L},L_{I},\hat{ \Sigma}, \delta^{\P}}$ where
 $\mathcal{L} = \powset(L)$ $ \setminus \{ \emptyset \}$, $L_{I} =\Supp(\mu_{0})$,
 $\hat{\Sigma}=\{\hat{\sigma} \mid \hat{\sigma}: L \to \Sigma\}$ is the alphabet, and $\delta^{\P} : \mathcal{L} \times \hat{\Sigma} \to \mathcal{L}$ where
 for all $s_{1}, s_{2} \in \mathcal{L}$ and $\hat{\sigma} \in \hat{\Sigma}$,
 we have $\delta^{\P}(s_{1}, \hat{\sigma})=s_{2}$ where $s_{2}=\cup_{\ell \in s_1} Post_{\hat{\sigma}(\ell)}(\ell)$.
\end{definition}

\begin{example}
\figurename~\ref{f5}(b) shows the  perfect information subset
construction $M^{\P}$ of the MDP drawn in \figurename~\ref{f5}(a)
(presented in the first example).
Let us present $\hat{\Sigma}$ in the table below. Each row
labelled by  a function $\hat{\sigma}_i$ ($i \in \{1,\dots,11\}$),
each column labelled by a state $\ell$; and each entry shows the
value of  $\hat{\sigma}_i(\ell)$.
\begin{displaymath}
\begin{array}{c|cccc}
                \text{}&1& 2& 3&4\cr
                \hline
                \hat{\sigma}_1&\sigma_2& \{\sigma_1,\sigma_2\}& \{\sigma_1,\sigma_2\}&\{\sigma_1,\sigma_2\}\cr
                \hat{\sigma}_2&\sigma_1& \{\sigma_1,\sigma_2\}& \{\sigma_1,\sigma_2\}&\{\sigma_1,\sigma_2\}\cr
                \hat{\sigma}_3&\{\sigma_1,\sigma_2\}& \sigma_2& \sigma_1&\{\sigma_1,\sigma_2\}\cr
                \hat{\sigma}_4&\{\sigma_1,\sigma_2\}& \sigma_1& \sigma_2&\{\sigma_1,\sigma_2\}\cr
                \hat{\sigma}_5&\{\sigma_1,\sigma_2\}& \sigma_2& \sigma_2&\{\sigma_1,\sigma_2\}\cr
                \hat{\sigma}_6&\{\sigma_1,\sigma_2\}& \sigma_1& \sigma_1&\{\sigma_1,\sigma_2\}\cr
                \hat{\sigma}_7&\{\sigma_1,\sigma_2\}& \{\sigma_1,\sigma_2\}& \{\sigma_1,\sigma_2\}&\{\sigma_1,\sigma_2\}\cr
                \hat{\sigma}_8&\{\sigma_1,\sigma_2\}&  \sigma_1 &\{\sigma_1,\sigma_2\}&\{\sigma_1,\sigma_2\}\cr
                \hat{\sigma}_9&\{\sigma_1,\sigma_2\}& \{\sigma_1,\sigma_2\}& \sigma_2&\{\sigma_1,\sigma_2\}\cr
                \hat{\sigma}_{10}&\{\sigma_1,\sigma_2\}& \{\sigma_1,\sigma_2\}& \sigma_1&\{\sigma_1,\sigma_2\}\cr
                \hat{\sigma}_{11}&\{\sigma_1,\sigma_2\}& \sigma_2& \{\sigma_1,\sigma_2\}&\{\sigma_1,\sigma_2\}
\end{array}%
\end{displaymath}
Note that, the function $\hat{\sigma}$ with
$\hat{\sigma}(\ell)=\{\sigma_1,\sigma_2\}$ (for a state $\ell$)
gives two different functions  where
$\hat{\sigma}_i(\ell)=\sigma_1$ and
$\hat{\sigma}_j(\ell)=\sigma_2$; but these two functions behaves
similarly.
\end{example}

A \emph{cycle} of $M^{\P}$ is a finite sequence
$C^{\P}=s_0~\hat{\sigma}_0 s_1 \dots s_{d-1} \hat{\sigma}_{d-1} s_d$ 
of interleaved cells and symbols such that  $\delta^{\P}(s_{j},
s_{j+1} = \hat{\sigma}_{j})$ for all $0 \leq j < d$, and $s_0=s_d$. 
Note that, in this definition, $d$ is the length of the cycle $C^{\P}$. We write $s \in C^{\P}$
if $s$ is one of the cells $s_j$ ($0 \leq j < d$) of the finite
sequence of the cycle $C^{\P}$. A \emph{simple cycle} is a cycle
where all cells $s_0, \dots , s_{d-1}$ are different.
We are interested in defining  some property on cycles  
of the perfect-information subset construction for a given MDP.


\begin{definition}[Recurrent cyclic sets]
Let $C^{\P}=s_0~\hat{\sigma}_0 ~\dots ~s_{d-1}~\hat{\sigma}_{d-1} s_d$ 
be a cycle  of  the perfect-information subset construction $M^{\P}$ 
for a given MDP $M$. A \emph{recurrent cyclic set} for the cycle  $C^{\P}$ is 
a sequence $G = g_0 g_1 \dots g_{d}$ such that $g_0 = g_d$, and
$\emptyset \neq g_i \subseteq s_i$ and $\cup_{\ell \in g_i} Post_{\hat{\sigma}_i(\ell)}(\ell) = g_{i+1}$
for all $0 \leq i < d$.
\end{definition}

A cycle $C^{\P}$ might have several recurrent cyclic sets.
A recurrent cyclic set $G$ for a given cycle $C^{\P}$, is said to be \textit{minimal} if there is no other recurrent cyclic set $G'$ ($G \neq G'$) such that
for $ 0 \leq i < d$, and for $g_i \in G$, $g'_i \in G'$, we have $g'_i \subseteq g_i$.
We denote the set of all minimal recurrent cyclic sets of the cycle $C^{\P}$
by $\Delta(C^{\P}) = \{G \mid G$ is a minimal recurrent cyclic set for the cycle $C^{\P}\}$.

\begin{example}
Consider the MDP~$M$ in \figurename~\ref{f7} (the initial
distribution is $\mu_0(1)=1$ and $\mu_0(i)=0$ for $i \in
\{2,\dots,9 \}$).
\figurename~\ref{f7}(b) shows one cycle of the perfect information subset
construction $M^{\P}$. Let us present $\hat{\Sigma}$ in the table below. Each row
labeled by  a function $\hat{\sigma}_i$ ($i \in \{1,\dots,4\}$),
each column labeled by a state $\ell$; and each entry shows the
value of  $\hat{\sigma}_i(\ell)$.
\begin{displaymath}
\begin{array}{c|ccccccccc}
                \text{}&1& 2& 3&4&5&6&7&8&9\cr
                \hline
                \hat{\sigma}_1& \{\sigma_1,\sigma_2\}& \{\sigma_1,\sigma_2\}& \{\sigma_1,\sigma_2\}&\{\sigma_1,\sigma_2\}& \{\sigma_1,\sigma_2\}& \{\sigma_1,\sigma_2\}& \{\sigma_1,\sigma_2\}& \{\sigma_1,\sigma_2\}& \{\sigma_1,\sigma_2\}\cr
                 \hat{\sigma}_2& \{\sigma_1,\sigma_2\}& \{\sigma_1,\sigma_2\}& \{\sigma_1,\sigma_2\}&\{\sigma_1,\sigma_2\}& \sigma_1& \{\sigma_1,\sigma_2\}& \{\sigma_1,\sigma_2\}& \{\sigma_1,\sigma_2\}& \{\sigma_1,\sigma_2\}\cr
                         \hat{\sigma}_3& \{\sigma_1,\sigma_2\}& \{\sigma_1,\sigma_2\}& \{\sigma_1,\sigma_2\}&\{\sigma_1,\sigma_2\}& \sigma_2& \{\sigma_1,\sigma_2\}& \{\sigma_1,\sigma_2\}& \{\sigma_1,\sigma_2\}& \{\sigma_1,\sigma_2\}\cr
                     \hat{\sigma}_4& \{\sigma_1,\sigma_2\}& \{\sigma_1,\sigma_2\}& \{\sigma_1,\sigma_2\}&\{\sigma_1,\sigma_2\}&\{\sigma_1,\sigma_2\}& \{\sigma_1,\sigma_2\}& \{\sigma_1,\sigma_2\}& \{\sigma_1,\sigma_2\}& \{\sigma_1,\sigma_2\}
\end{array}%
\end{displaymath}

For the cycle $C^{\P}=\{2,5,8\}~\hat{\sigma}_2$ $\{3,5,6\}~\hat{\sigma}_3$ $\{4,7,9\}$ $\hat{\sigma}_4~\{2,5,8\}$,
the set of minimal recurrent cyclic sets is 
$\Delta(C^{\P})=\{\{\{2\},\{3\},\{4\} \},\{\{5\},\{6\},\{7\} \}\}$.
The elements of $\Delta(C^{\P})$ are not comparable.
\end{example}

\begin{figure}[t]
\begin{center}
\includegraphics[scale=0.5]{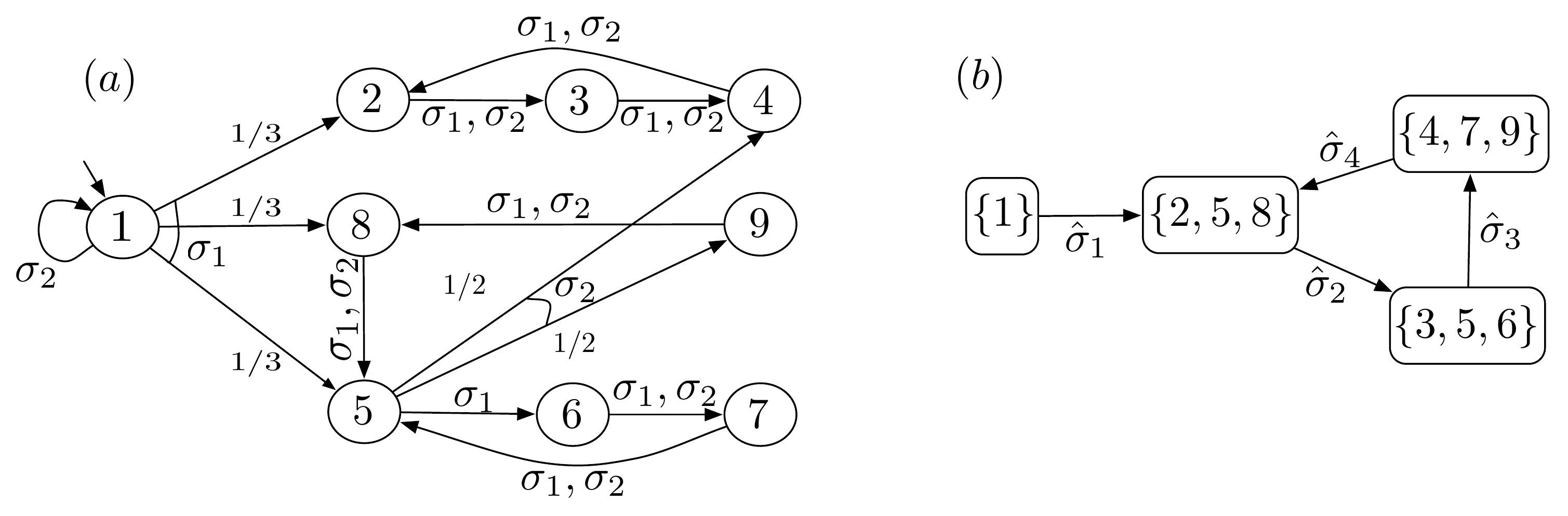}
\caption{(a) shows an MDP, and (b) shows some part of its perfect information subset
construction.}\label{f7}
\end{center}
\end{figure}


The blind subset construction for an MDP is a special case of its
perfect information subset construction where the action functions 
$\hat{\sigma} \in \hat{\Sigma}$ are restricted to constant functions. 
In each cell, all states have to choose the same action.
\\
\begin{definition}[Blind subset construction of an MDP]
The blind subset construction for a given MDP $M
=\tuple{L,\mu_{0},\Sigma,\delta}$ is an automaton
$M^{\B}=\tuple{\mathcal{L},L_{I}, \Sigma, \delta^{\B}}$ where
 $\mathcal{L} = \powset(L) \setminus \{ \emptyset \}$, $L_{I} =\Supp(\mu_{0})$,
 and for all $s_{1}, s_{2} \in \mathcal{L}$ and $\sigma \in \Sigma$,
 we have $\delta^{\B}(s_{1}, \sigma)=s_{2}$ where $s_2 = Post_{\sigma}(s_{1})$.
\end{definition}

We denote cycles in the blind subset construction by $C^{\B}$.


\section{Synchronizing Objectives for Perfect-Information Strategies}
We have defined a perfect-information one-player stochastic game in which the
player can see the current state of the game and record the
sequence of visited states. We show that synchronizing strategies
can be characterized in the perfect-information subset
construction, giving a decidability result. We also show in the
next example that memory may be necessary.


\begin{example}
Consider the MDP~$M$ in \figurename~\ref{f6} (the initial
distribution is $\mu_0(1)=1$ and $\mu_0(i)=0$ for $i \in
\{2,\dots,5 \}$), and let $\alpha$ be the strategy defined as
follows: $\alpha((L \times \Sigma)^{*}\ell)(\sigma)=1/2$ for all
$\sigma \in \Sigma$ and $\ell \in \{1,3,4,5\}$, and for the
histories ending in the state $2$,
\begin{displaymath}
 \alpha((L \times \Sigma)^{*}\ell~\Sigma~2)(\sigma)= \left\{
     \begin{array}{ll}
       1& $if$~\ell=1 $ and $ \sigma=\sigma_2,\\
       1 &$if$~\ell\neq 1 $ and $ \sigma=\sigma_1,\\
       0& $otherwise$.
     \end{array}
   \right.
\end{displaymath}
\end{example}


\begin{figure}[t]
\begin{center}
\includegraphics[scale=0.5]{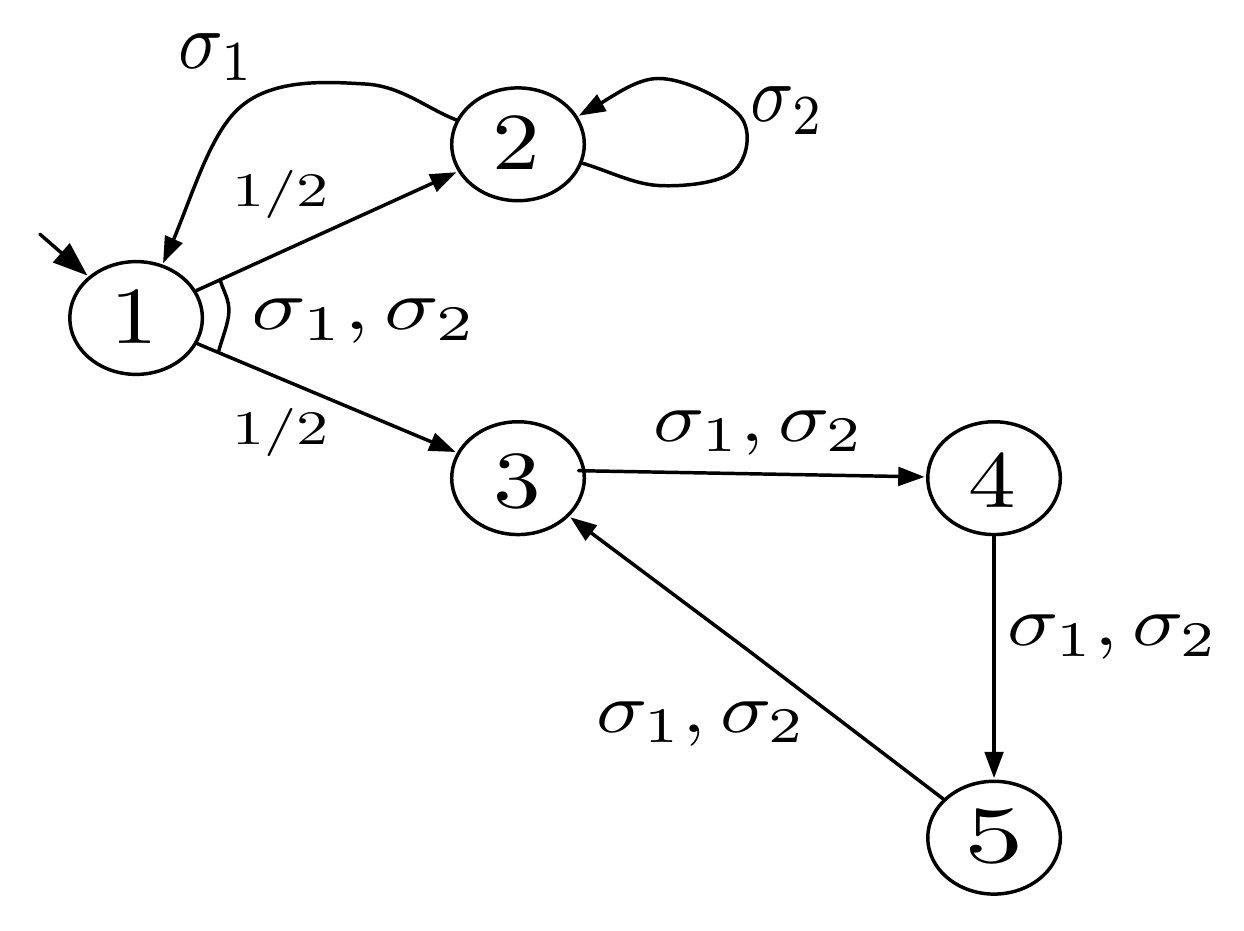}
\caption{An MDP where memory is necessary to win the strongly
synchronizing objective.}\label{f6}
\end{center}
\end{figure}


\indent In this example, it is easy to check that the strategy
$\alpha$ is strongly synchronizing. In state $2$, it plays
$\sigma_1$ and $\sigma_2$ in alternation in order to ensure
synchronization with the cycle $3,4,5$ of length $3$. However,
none of the memoryless strategies is strongly synchronizing,
showing that memory is necessary. This example also shows that
memory is necessary for weakly synchronizing objective, as well as for
blind strategies.


\begin{proposition}
 For both strongly and weakly synchronizing objectives, memoryless
strategies are not sufficient in MDPs.
\end{proposition}


\begin{theorem}\label{theo-perf-col-r1}
For a perfect information game over an MDP~$M$, there exists a
strategy~$\alpha$ such that $M$ with strategy~$\alpha$
 is \textbf{strongly synchronizing},
if and only if the perfect-information subset construction
$M^{\P}$ for~$M$, has an accessible  cycle $C^{\P}$ such
that $\abs{\Delta(C^{\P})}=1$, and for $G \in \Delta(C^{\P})$ and
for all $g \in G$, $\abs{g}=1$.
\end{theorem}


\begin{proof} \emph{Sufficient condition.}
We suppose that the
perfect-information subset construction
$M^{\P}$ for
$M$, has an accessible  cycle
$C^{\P}=s_0~\hat{\sigma}_0 ~\dots  s_d$ such
that $\abs{\Delta(C^{\P})}=1$, and for $G \in \Delta(C^{\P})$ and
for all $g \in G$, we have $\abs{g}=1$.
Since this cycle is accessible, there exists a finite path
$P=p_0 \hat{\sigma}'_0 p_1 \dots p_{m-1} \hat{\sigma}'_{m-1} p_{m}$ in $M^{\P}$
from $ p_0 = L_I$ to $ p_{m}=s_0=s_d$ (See \figurename~\ref{f1}).
Consider the pure strategy $\alpha$ as follows
\begin{displaymath}
  \alpha((L \times \Sigma)^k \ell) = \left\{
     \begin{array}{ll}
        \hat{\sigma}'_k(\ell) & $if$~ 0 \leq k < m,\\
      \hat{\sigma}_{(k-m)~mod~d}(\ell) & $if$~ m \leq k.
     \end{array}
   \right.
\end{displaymath}


\indent Let us construct a finite Markov chain $M'$ in a way that
its long term behavior simulates the long term behavior of the MDP
$M$ under the strategy $\alpha$ for synchronizing objectives.
This Markov chain is $M'=(L',\mu'_0, \delta')$
where $L'=\{(i,\ell)
\mid 0 \leq i < (m+d) \mbox{ and } \ell \in L\}$,
the initial distribution $\mu'_0$ is defined as follows

\begin{displaymath}
 \mu'_0((i,\ell))= \left\{
     \begin{array}{ll}
       \mu_0(\ell) & $ if $ i=0 \\
0 & otherwise.
     \end{array}
   \right.
\end{displaymath}

\noindent and the probability transition function $\delta'$ is
defined as follows
\begin{displaymath}
  \delta'((i,\ell))((i',\ell'))= \left\{
     \begin{array}{ll}
     \delta(\ell,\hat{\sigma}'_i(\ell))(\ell') &$if$~ (0 \leq i < m),~(i'=i+1),~ (\ell \in p_i) $ and $ (\ell' \in p_{i'}),\\
      \delta(\ell,\hat{\sigma}_{i-m}(\ell))(\ell') &$if$~ (m \leq i < m+d) ,~(i'=m+(i-m+1)~mod~d),\\
      &\hfill ~ (\ell \in s_{i-m}) $ and $ (\ell' \in s_{i'-m}),\\
0 & otherwise.
     \end{array}
   \right.
\end{displaymath}

The idea is that each cell $p_i$ ($0 \leq i < m$) of the path $P$
and, similarly, each cell $s_i$ ($m \leq i < m+d$) of the cycle
$C^{\P}$ corresponds to $\abs{L}$ states in the Markov chain $M$ (one for each state of the MDP $M$).
 The value of $\delta'((i,\ell))((i',\ell'))$ shows the
probability to reach in one step, the state $(i',\ell')$ from the
state $(i,\ell)$; semantically it gives the probability to go from
$\ell$ to $\ell'$ at step $i$.
We show that (a)  if the Markov chain $M'$ is strongly synchronizing,
then so is the MDP $M$ under the strategy $\alpha$ and that (b)
$M'$ is strongly synchronizing.


\begin{figure}[t]
\begin{center}
\includegraphics[scale=0.5]{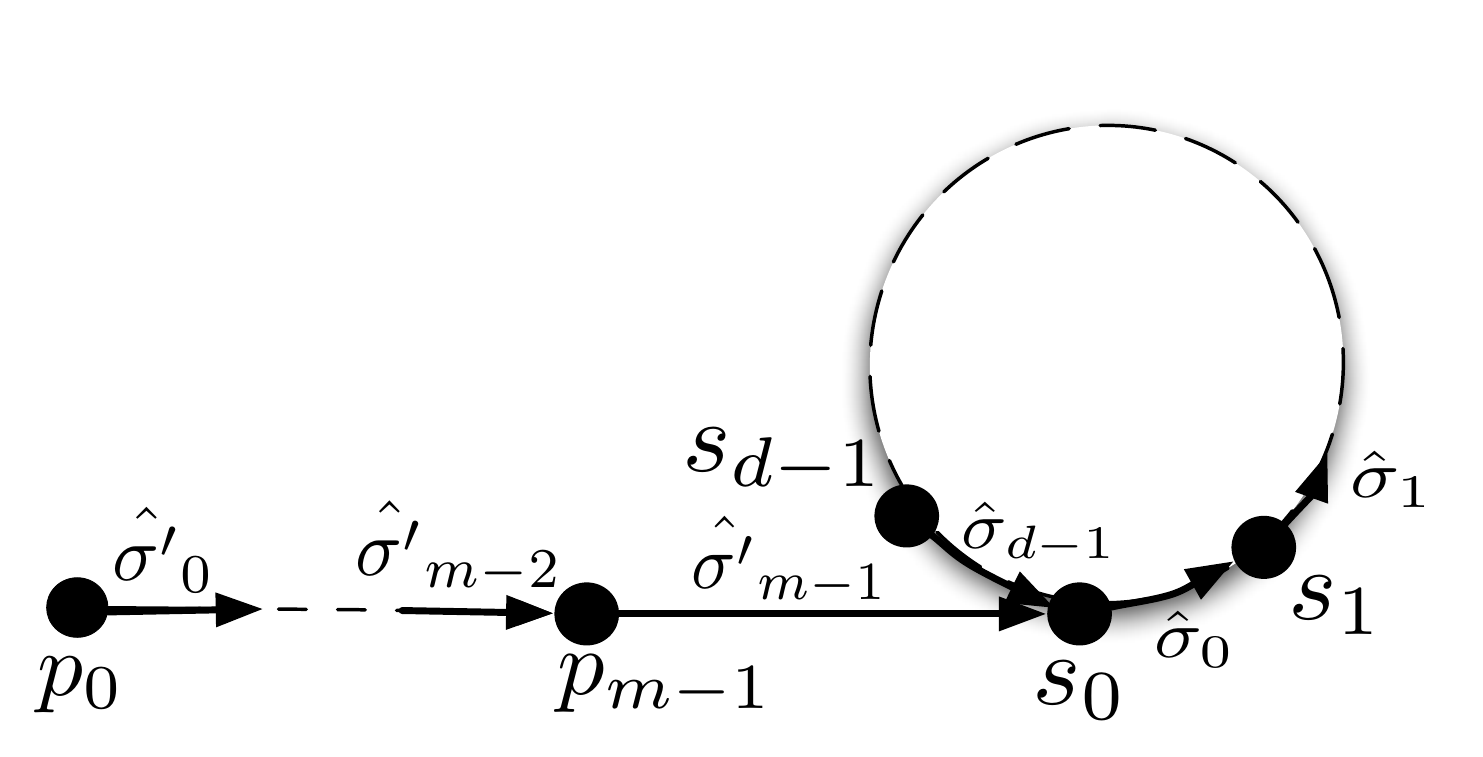}
\caption{An accessible cycle $C^{\P}$ of $M^{\P}$ which is
reachable by a finite path $p_0, \dots, p_{m}$.}\label{f1}
\end{center}
\end{figure}

Proving (a) is straightforward from the definition of the Markov chain $M'$. Each
state of the MDP $M$ corresponds to $m+d$ state of $M'.$
Then if,
from some point,
the mass of probability accumulates in one state of
$M'$ and afterward moves totally to another one,
it happens also in $M$.
In detail,
let the sequence $X^{\alpha}_i$ ($i \in \nat$) denote the outcome of the MDP $M$ under the strategy $\alpha$,
and  $X'_i$ ($i \in \nat$) denote the probability distribution at step $i$ generated by the Markov chain $M'$.
Note that  $X^{\alpha}$ is a random variable over $\abs{L}$ entries,
but  $X'$ is over $\abs{L}\cdot(m+d)$ entries which has at most $\abs{L}$ non-zero entries. Let us compute and compare the non-zero entries of these two random variable sequences. For $\ell \in L$:
\begin{flushleft}
$X^{\alpha}_0(\ell)=\mu_0(\ell)=X'_{0}((0,\ell))$ and we have $X'_{0}((j,\ell))=0$ for all $j \neq 0$.
\end{flushleft}
\begin{flushleft}
    $X^{\alpha}_1(\ell)=
\sum_{\ell' \in L}\mu_0(\ell')\cdot \delta (\ell',\alpha(\ell'))(\ell)=
\sum_{\ell' \in L}\mu_0(\ell')\cdot \delta (\ell',\hat{\sigma}'_0(\ell'))(\ell)=   \sum_{\ell' \in L}\mu_0(\ell')\cdot \delta'((0,\ell'))((1,\ell))=
X'_{1}((1,\ell))$ and we have $X'_{1}((j,\ell))=0$ for all $j \neq 1$.
\end{flushleft}
In the next step, let us compute these random variables for $i<m$:

\begin{flushleft}
$X^{\alpha}_i(\ell)=
\sum_{\ell_0,\ell_1,\dots \ell_{i-1} \in L}\mu_0(\ell_0)\cdot \delta (\ell_0,\alpha(\ell_0))(\ell_1)\cdot \delta (\ell_1,\alpha(\ell_0 ~\alpha(\ell_0)~\ell_1))(\ell_2) \dots \cdot \delta (\ell_{i-1},\alpha(\ell_0 ~\alpha(\ell_0)~\ell_1\dots \ell_{i-1}))(\ell)=
 \sum_{\ell_0,\ell_1,\dots \ell_{i-1} \in L}\mu_0(\ell_0)\cdot\delta (\ell_0,\hat{\sigma}'_0(\ell_0))(\ell_1)\cdot\delta (\ell_1,\hat{\sigma}'_1(\ell_1))(\ell_2)\dots \cdot\delta (\ell_{i-1},\hat{\sigma}'_{i-1}(\ell_{i-1}))(\ell)
=\sum_{\ell_0,\ell_1,\dots \ell_{i-1} \in L}\mu_0(\ell_0) \cdot\delta'((0,\ell_0))((1,\ell_1)) \cdot \delta'((1,\ell_1))((2,\ell_2)) \dots
 \cdot \delta'((i-1,\ell_{i-1}))((i,\ell)) =
X'_{i}((i,\ell))$.
\end{flushleft}

We, also, have $X'_{i}((j,\ell))=0$ for all $j \neq i$, these results give $\norm{X^{\alpha}_i}=\norm{X'_i}$ for $i<m$.
At the end, consider  $i \geq m$:

\begin{flushleft}
$X^{\alpha}_i(\ell)=
\sum_{\ell_0,\ell_1,\dots \ell_{i-1} \in L}\mu_0(\ell_0)\cdot \delta (\ell_0,\alpha(\ell_0))(\ell_1)\cdot \delta (\ell_1,\alpha(\ell_0 ~\alpha(\ell_0)~\ell_1))(\ell_2) \dots \cdot \delta (\ell_{i-1},\alpha(\ell_0 ~\alpha(\ell_0)~\ell_1\dots \ell_{i-1}))(\ell)=
\sum_{\ell_0,\ell_1,\dots \ell_{i-1} \in L}\mu_0(\ell_0)\cdot\delta (\ell_0,\hat{\sigma}'_0(\ell_0))(\ell_1)\cdot\delta (\ell_1,\hat{\sigma}'_1(\ell_1))(\ell_2)\dots \cdot\delta (\ell_{m-1},\hat{\sigma}'_{m-1}(\ell_{m-1}))(\ell_{m})
\cdot\delta (\ell_{m},\hat{\sigma}_{0}(\ell_{m}))(\ell_{m+1})
\dots \cdot\delta (\ell_{i-1},\hat{\sigma}_{(i-m)~mod~d}(\ell_{i-1}))(\ell)
=\sum_{\ell_0,\ell_1,\dots \ell_{i-1} \in L}\mu_0(\ell_0) \cdot\delta'((0,\ell_0))((1,\ell_1)) \cdot \delta'((1,\ell_1))((2,\ell_2)) \dots
 \cdot \delta'((m-1,\ell_{m-1}))((m,\ell_{m})) \dots
 \cdot \delta'((m+(i-m)~mod~d,\ell_{i-1}))((m+(i-m)~mod~d,\ell)) =
X'_{i}((m+(i-m)~mod~d,\ell))$.
\end{flushleft}

We, also, have $X'_{i}((j,\ell))=0$ for all $j \neq i$, this results give $\norm{X^{\alpha}_i}=\norm{X'_i}$ for  $i \geq m$.
We have shown that $X^{\alpha}_i(\ell)=X'_i((j,\ell))$ where for $0 \leq i <m$, we have $j=i$, and for $i \geq m$, we have $j=m+(i-m)~mod~d$.
This simply gives $\norm{X^{\alpha}_i}=\norm{X'_i}$ for $i \in \nat$;
meaning that if the Markov chain $M'$ is synchronizing,
so is the MDP $M$ under the strategy $\alpha$.

To show (b), we study transient and recurrent states of the Markov chain
$M'$.
Suppose that $G \in \Delta(C^{\P})$ is the only recurrent cyclic set of the cycle, and it includes $d$ elements as $g_0, \dots g_{d-1}$.
Let $R$ be the set of states  $(m+i,\ell)$ such that $\ell \in g_i$, for $0 \leq i < d$.
We claim that the states of $R$ are the only recurrent states in the Markov chain $M'$.

\begin{itemize}
    \item First, we can see that the states of $R$ are  recurrent.
        By construction, the states of $R$ are strongly connected.
In addition, we have to prove that if $(m+i,\ell) \in R$ and $(m+i,\ell) \rightarrow (m+j, \ell')$, then $(m+j, \ell') \in R$. This holds by induction on the equality $\cup_{\ell \in g_i}
Post_{\sigma_i(\ell)}(\ell) = g_{i+1}$. Note that $(m+i,\ell) \in R$ implies that $\ell \in g_i$; and if $(m+i,\ell) \rightarrow (m+j, \ell')$ then $\ell'$ has to lie in $g_j$.

    \item Now, we  show that the states of $R$ are the \textbf{only} recurrent      states.
    By contradiction, suppose that there is another set $R'$ of recurrent states in the Markov chain $M'$.
By Proposition \ref{proposition-mc} and since the states $(i,\ell)$ ($0 \leq i < m$) are visited only once, then they could not be recurrent;
therefore we discuss only on the states  $(m+i,\ell)$  with $0 \leq i < d$ of the Markov chain $M'$.
Let $g'_i$ denote all states included in $\{\ell \mid (m+i,\ell) \in R'\} \cap s_i$ for $0 \leq i < d$.
The construction of the Markov chain implies that a state $(m+i,\ell)$ can only have outgoing edges toward some states $(m+(i+1) ~mod~d,\ell')$;
hence  $g'_i \neq \emptyset$ for all $0 \leq i < d$.
On the other hand, the definition of recurrent states requires that each accessible state from $(m+i,\ell) \in R'$ could access $(m+i,\ell)$,  therefore $ \cup_{\ell \in g'_i}Post_{\sigma_i(\ell)}(\ell) =g'_{i+1}$.
It is a contradiction with $\abs{\Delta(C^{\P})}=1$.
\end{itemize}

\indent Based on Proposition \ref{proposition-mc}, for the transient states $(k,\ell)$,
the probability $X_n((k,\ell))$ vanishes for $n \to \infty$.
Since for all $g \in G$, we have $\abs{g}=1$,
the support of $X_n$ ($n>m$) contains only one recurrent state.
Thus, the probability mass accumulates in that state:
for all $\epsilon>0$, for all $n>n_0$ there is a state $(i,\ell)$ with
$X_n((i,\ell))> 1-\epsilon$, that is $\norm{X^{\alpha}_n} >
1-\epsilon$.
Hence, $\lim_{n \to \infty} \norm{X^{\alpha}_n} = 1$
and~$M'$ is strongly synchronizing.
Therefore, so is the MDP $M$ under the strategy $\alpha$.
\\
\\
\noindent \emph{Necessary condition.} Assume that the MDP $M$ with
strategy $\alpha$ is strongly synchronizing.
Then $\forall \epsilon>0 \cdot  \exists n_0 \in \nat \cdot \forall n \geq n_0 \cdot \exists\q_n$ 
such that $X^{\alpha}_{n}(\q_n)> 1-\epsilon$.
Moreover the state $\q_n$ is unique, and we show below that it is independent 
of~$\epsilon$ (assuming $\epsilon < \frac{1}{2}$).

Let $\nu$ be the smallest probability among all probability distributions of the MDP $M$ (i.e., $\nu=\min_{\ell \in L, \sigma \in \Sigma,\ell' \in \Supp(\delta(\ell,\sigma))}(\delta(\ell,\sigma)(\ell'))$).
Let $\epsilon < \frac{\nu}{1+\nu}$.
We claim that for all $n \geq n_0$,
there exists some action $\sigma \in \Sigma$
such that $Post_{\sigma}(\q_n)=\{\q_{n+1}\}$ is a singleton.
Toward contradiction, assume that
for all $\sigma \in \Sigma$,
the statement $Post_{\sigma}(\q_n) \neq \{\q_{n+1}\}$ is satisfied.
The   probability which does not inject to $\q_{n+1}$ (from $\q_n$) is at least $\nu \cdot (1-\epsilon)$. And since $M$ is strongly synchronizing, we have:
\begin{center}
$1- \epsilon \leq  \norm {X^{\alpha}_{n+1}} \leq 1-\nu \cdot (1-\epsilon)$
\end{center}
 This  gives $\epsilon \geq \frac{\nu}{1+\nu}$ which is a contradiction.
Therefore, for all $n \geq n_0$, there exists $\sigma \in \Sigma$
such that
 $Post_{\sigma}(\q_n)=\{\q_{n+1}\}$.
This implies that the infinite sequence of states $I=\q_{n_0} \q_{n_0+1} \dots$ 
is uniquely defined.

The sequence $I$ is used to define a pure synchronizing strategy $\beta$ 
from the randomized synchronizing strategy~$\alpha$. This construction 
implies that the pure strategies are sufficient for strongly synchronizing objectives.
We define the pure strategy $\beta$ as follows:
\begin{itemize}
    \item for $h \in \Hists (M)$ with $\abs{h}=i$ and  $\Last(h) =\q_i$, we define $\beta(h)=\sigma$ where $Post_{\sigma}(\Last(h))=\{\q_{i+1}\} $,
\item for $h \in \Hists (M)$ with  $\abs{h}=i$ and $\Last(h) \neq \q_i$,
we define $\beta(h)= \Action(h',i)$

where $h'\in \Hists(M)$ is the shortest possible history such that 
(1)~$\State(h',i) = \Last(h)$,
(2)~$Pr^{\alpha}(h')$ $>0$, and  
(3)~$\Last(h')=\q_j$ with $\abs{h'}=j$.
One might notice that a reachable state $\Last(h)$ with a strictly positive 
probability $Pr^{\alpha}(h) >0$, has to access a state of $I$ 
(such as $\Last(h')=\q_j$ where $\abs{h'}=j$); 
otherwise the MDP $M$ with strategy $\alpha$, 
would not be strongly synchronizing. 
Consequently, the history $h'$ defined above always exists.
\end{itemize}

As a result, we can define SizePath$(h)=\abs{h'}-\abs{h}$ to be the size of 
shortest path from $\Last(h)$ to  the infinite sequence $I$.
Note  that for $h$  with $\abs{h}=i$ and  $\Last(h) =\q_i$, we define SizePath$(h)=1$.
It is easy to see that  the MDP $M$ with pure strategy $\beta$ is also strongly synchronizing.

In the following, we show that there exists a 
cycle  $C^{\P}$ of $M^{\P}$ which has only one recurrent cyclic set $G$, 
and all $g \in G$ are singleton.
By construction, we have $\beta(h)=\beta(h')$ for all histories $h,h' \in \Hists(M)$ 
with $\Last(h)=\Last(h')$ and $\abs{h}=\abs{h'}$. Therefore the pure strategy 
$\beta$ induces an infinite path $P_{\beta},$ in the perfect-information 
subset construction~$M^{\P}$.
Since the state space of $M^{\P}$ is finite, some cell $S$ has to be visited 
infinitely many times along~$P_{\beta}$.
The path between two visits to $S$ along $P_{\beta}$ is a cycle (not 
necessarily a simple cycle) of ~$M^{\P}$.
We study one of the these cycles (starting at $S$ and coming back there), 
and prove that this cycle satisfies the conditions of the theorem.

Let $\Inf(I)$ denotet the set of all states visited infinitely often along $I$.
Hence, there exists $N_{\Inf} \geq n_0$ such that $\forall i\geq N_{\Inf}: \q_{i} \in I \Rightarrow \q_{i} \in \Inf(I)$.
Let  $K_{1}$ be the first step after $N_{\Inf}$ in which  the path $P_{\beta}$ visits~$S$.
Let \textit{MaxPath}=$\max_{h \in \Hists(M), Pr^{\beta}(h)>0, \abs{h}= K_{1}}($SizePath$(h))$, 
 be the length of the longest path (among the shortest ones) from one reachable state 
at step $K_{1}$, to the infinite sequence $I$.

Let $C^{\P}$ be the cycle starting in $S$ at step $K_1$, and coming back to this 
state in some step $K_2>K_1+$\textit{MaxPath}.
We claim that this cycle $C^{\P}$ has only one recurrent cyclic set
$G$, and all subsets  $g \in G$ are singleton:
\begin{enumerate}
    \item $G=\{\{\q_i \} \mid \q_{i} \in I $ for $K_1 \leq i \leq K_2\}$ is a   
recurrent cyclic set. We already have proved that  there exists $ \sigma \in  \Sigma$ such that 
$Post_{\sigma}(\q_n)=\{\q_{n+1}\}$ ($n \geq n_0$). Note that for state $q_n$, 
the action $\sigma$ is chosen by the cycle. 
    \item $G$ is the only recurrent  cyclic set. Each state included in $S$ 
reaches, in at most \textit{MaxPath} steps,  one state of $I$. Hence, 
the cell $S$, as the first element of $C^{\P}$, cannot have another subset $g'$ 
constructing another  recurrent cyclic set.
\end{enumerate} 
We have proved that for a strongly synchronizing  MDP $M$, the perfect information 
subset construction for $M$, has a cycle $C^{\P}$ such
that $\abs{\Delta(C^{\P})}=1$, and for $G \in \Delta(C^{\P})$ and
for all $g \in G$, $\abs{g}=1$.
\begin{flushright}
$\Box$\end{flushright}
\end{proof}
Through the proof of Theorem \ref{theo-perf-col-r1}, we have seen
that for all strategies $\alpha$ such that an MDP $M$ with the
strategy $\alpha$ is strongly synchronizing, there is a pure
strategy that satisfies the strongly synchronization condition. We
will see that this is also the case for weakly synchronizing
objective (see the proof of Theorem \ref{theo-perf-col-r2}).


\begin{corollary}
 For both strongly and weakly synchronizing objectives, pure
strategies are sufficient in MDPs.
\end{corollary}


\begin{theorem}\label{theo-perf-col-r2}

For a perfect information game over an MDP~$M$, there exists a
strategy~$\alpha$ such that $M$ with strategy~$\alpha$
 is \textbf{weakly synchronizing},
if and only if the perfect-information subset construction
$M^{\P}$ for $M$, has an accessible  cycle $C^{\P}$ such
that $\abs{\Delta(C^{\P})}=1$, and for $G \in \Delta(C^{\P})$,
there exists $g \in G$ such that~$\abs{g}=1$.

\end{theorem}

\begin{proof}
\noindent \emph{Sufficient condition.} We suppose that the
perfect-information subset construction
$M^{\P}$ for
$M$, has an accessible cycle $C^{\P}$  such
that $\abs{\Delta(C^{\P})}=1$, and for $G \in \Delta(C^{\P})$,
there exists $g \in G$ such that $\abs{g}=1$. Consider a pure
strategy similarly to which presented in proof of Theorem
\ref{theo-perf-col-r1}. Let us, here as well, construct the Markov
chain $M'$, and therefore discuss on transient and recurrent
states of $M'$.

Suppose that $G \in \Delta(C^{\P})$ is the only recurrent cyclic set of the cycle, and it includes $d$ elements as $g_0, \dots g_{d-1}$.
Let $R$ be the set of states  $(m+i,\ell)$ such that $\ell \in g_i$, for $0 \leq i < d$.
As we have shown in proof of Theorem
\ref{theo-perf-col-r1}, the states of $R$ are the only recurrent states in the Markov chain $M'$.
Let $p_n$ be the probability to be in one
state of $R$ at step $n$.
 Based on Proposition \ref{proposition-mc}, for the transient
states $(i,\ell)$ the probability $X_n((i,\ell))$ vanishes for $n \to \infty$, 
which leads to $lim_{n\to \infty} \, p_n=1$.
On the other hand, by hypothesis, for $G \in \Delta(C^{\P})$
there exists $g_j \in G$ ($ 0 \leq j <d$) such that $\abs{g_j}=1$.
Then every $d$ steps, at least once, the probability $p_{m+ k\cdot d + j}$ 
gathers in only one state $(m+j,\ell)$ where $\ell \in g_{j}$.
As a result, for all $k \in \mathbb{N}$,
$\max(\norm{X_{m+d.k}^{\alpha}},\norm{X_{m+d.k+1}^{\alpha}},...\norm{X_{m+d.k+d-1}^{\alpha}}) \geq p_{m+ k\cdot d + j}$. 
We have shown that $lim_{n\to \infty} \, p_n=1$, hence
$limSup_{n \to \infty} \, \norm{X_{n}^{\alpha}}=1$.
\\
\\
\noindent \emph{Necessary condition.}
Assume that the MDP $M$ with
strategy $\alpha$ is weakly synchronizing meaning that 
$limSup_{n \to \infty} \, \norm{X_{n}^{\alpha}}=1$.
Therefore there exists a subsequence $\norm{X^{\alpha}_{i_k}}$ of 
$\norm{X^{\alpha}_{i}}$ which approaches to~$1$ (i.e., $lim_{k \to \infty} \, \norm{X_{i_k}^{\alpha}}=1$), 
where  $i_0 < i_1 < i_2 < \dots$ is an increasing sequence of indices. 
Then, for  $\epsilon < 1/2$ there exists $n_0 \in \nat$ such that for all $n \geq n_0$ 
there exists a (unique) state $\ell$ such that $X^{\alpha}_{i_n}(\ell) > 1/2$.
Let $(\ell,i_n)$ refer to this unique state at position $i_n$. 
Let $\Inf$ be the set of all states $\ell$ such that 
$X^{\alpha}_{i_n}((\ell,i_n)>1/2$ for infinitely many $n \in \nat$. 

Hence, there exists $N_{\Inf} \geq n_0$ such that 
$\forall n\geq N_{\Inf}: X^{\alpha}_{i_n}((\ell,i_n)>1/2 \Rightarrow \ell \in \Inf$. 

Since the state space of the MDP is finite, for a specific $\q \in \Inf$, 
we can define a subsequence $(j_k)_{k \in \nat}$) of $(i_k)_{k \in \nat}$ such that
\begin{enumerate}
   \item $j_0 \geq N_{\Inf}$, and
   \item $X^{\alpha}_{j_k}((\q,j_k)>1/2$, and
   \item $\Supp(X^{\alpha}_{j_k})=\Supp(X^{\alpha}_{j_{k+1}})$; 
in the sequel, we denote to this set by $S$.
\end{enumerate}  

 Let $(\q,j_k)$ refer to the state $\q$ at specific step $j_k$, and $J$ be the 
sequence of this states.  Note that since  $j_k$ is a subsequence of  $i_k$, we have 
$lim_{k \to \infty} \, \norm{X_{j_k}^{\alpha}}=1$ as well. 

We use the infinite sequence  $J$  to construct a winning pure strategy from the
winning randomized strategy $\alpha$.  
Consider the pure strategy $\beta$ as follows.
For $h \in \Hists (M)$ with  $\abs{h}=i$,
we define $\beta(h)= \Action(h',i)$
where

$h'\in \Hists(M)$ is the shortest possible history such that 
(1)~$Pr^{\alpha}(h') >0$, 
(2)~$\Last(h')= (\q,j_k)$  where $\abs{h'}=j_k$ for some $k \in \nat$, and in addition
(3)~$\State(h',i)=\Last(h)$.
One might notice that a reachable state $\Last(h)$ with a strictly positive 
probability $Pr^{\alpha}(h) >0$, has to access the infinite sequence  $J$; 
otherwise the MDP $M$ with strategy $\alpha$ 
would not be weakly synchronizing. Consequently, the history $h'$ defined above 
always exists.

Similarly to the case of strongly synchronizing, we can define SizePath$(h)=\abs{h'}-\abs{h}$ 
to be the size of shortest path from $\Last(h)$ to the  infinite sequence  $J$.

In the following, we show that for a weakly synchronizing MDP $M$, there exists a 
cycle  $C^{\P}$ of $M^{\P}$ which has only one recurrent cyclic set $G$, and there 
exists $g \in G$ which is singleton.
By construction, we have $\beta(h)=\beta(h')$ for all histories $h,h' \in \Hists(M)$ 
with $\Last(h)=\Last(h')$ and $\abs{h}=\abs{h'}$. Therefore the pure strategy $\beta$ 
induces an infinite path $P_{\beta}$ in the perfect-information subset 
construction~$M^{\P}$.
The construction of $\beta$, also implies that the cell $S$ is  visited infinitely 
many times along~$P_{\beta}$.
The path taken between two visits to $S$ along $P_{\beta}$ is a cycle (not 
necessarily a simple cycle) of~$M^{\P}$.
We study one of  these cycles (starting at $S$ and coming back there), 
and prove that this cycle satisfies the conditions of the theorem.

Let $K_{1}$ to be the first step after $N_{\Inf}$ in which  the path $P_{\beta}$ visits $S$.
Let us define \textit{MaxPath} = $\max_{h \in \Hists(M), Pr(h)>0, \abs{h}= K_{1}}($SizePath$(h))$
to be the length of the longest path (among the shortest ones) from a reachable 
state at step $K_{1}$ to the infinite sequence~$J$.

Let $C^{\P}$ be the cycle starting in $S$ at step $K_1$, and coming back to this 
state in some step $K_2>K_1+$\textit{MaxPath}. For convenience, let $d=K_2 -K_1$ 
denote the length of the cycle $C^{\P}$.
We define the winning pure strategy $\beta'$ from the strategy $\beta$ as follows.

\begin{itemize}
  \item for $h\in \Hists(M)$ with $\abs{h} <K_{1}+K_{2}$, we define $\beta'(h)=\beta(h)$.
  \item for $h\in \Hists(M)$ with $\abs{h}>K_{1}+K_{2}$, we define $\beta'(h)=\beta(h')$ 
where $\abs{h}=d\cdot m+ \abs{h'}$ for some $m \in \nat$, 
and $h'$ is a history with $ K_{1} \leq \abs{h'}\leq K_{1}+K_{2} $ and $\Last(h)=\Last(h')$. 
\end{itemize}

In fact, the path corresponding to the strategy $\beta'$ first reaches the cycle $C^{\P}$, 
and then forever follows this cycle.
The strategy $\beta'$ as well as the strategy $\beta$ is weakly synchronizing.  
We claim that this cycle $C^{\P}=s_0 ~\hat{\sigma}_0~ \cdots~ s_d$  ($s_0=s_d$) 
has only one recurrent cyclic set
$G$, and there exists  $g \in G$ which is singleton:
\begin{enumerate}
    \item First we prove that this cycle has one recurrent cyclic set.
	       The size of the cycle is more than \textit{MaxPath} which shows 
that some elements of the  infinite sequence  $J$  are visited along the cycle.
Suppose that  $(\q,j_{k'})$ is the last visited state of $J$ along the cycle, 
and $K'=j_{k'}-K_1$ is the index of cell $s_{K'}$ including this state.
Let us construct a singleton subset $g_{K'}=\{(\q,K')\}$.
By induction, let $g_{(K'+i+1)~ mod ~d}=\cup_{\ell \in g_{(K'+i)~ mod ~d} } Post_{\hat{\sigma}_{(K'+i)~ mod ~d}}(\ell)$ 
for all $0 \leq i < d$.
Note that, for $i=d-1$, the set $g_{K'}$ is computed.
By definition, the set $G=\{g_0,~g_1,~\cdots,~g_{d-1}\}$ is a recurrent cyclic set, 
if after the computation, we still have $g_{K'}=\{(\q,K')\}$.

We claim that  $g_{K'}=\{(\q,K')\}$.
By contradiction, suppose that $g_{K'} \neq \{(\q,K')\}$ is satisfied. 
We have $\limsup_{n \to \infty}\norm{X^{\beta'}_{n}}=1$.  
Then $\forall \epsilon>0 \cdot  \exists n_0 \in \nat \cdot \forall n \geq n_0 \cdot \exists \ell$ 
such that $X^{\beta'}_{n}(\ell)> 1-\epsilon$. On the other hand, by definition of $J$,
we know that the mass of probability in states of $J$ are more than $1/2$, 
and in addition we know that all states of the cycle inject probability to $J$; 
these show that the visited states of $J$ along the cycle are candidates to 
concentrate the probability mass. Let $\nu$ be the smallest probability among 
all probability distributions of the MDP $M$ 
(i.e., $\nu=\min_{\ell \in L, \sigma \in \Sigma,\ell' \in \Supp(\delta(\ell,\sigma))}(\delta(\ell,\sigma)(\ell'))$).
Let $\epsilon < \frac{\nu^{d}}{1+\nu^{d}}$,
and $X^{\beta'}_{n}((\q,K'))> 1-\epsilon$ where $n>n_0$.
The  probability which does not inject to $(\q,K')$ (from  $(\q,K')$ after $d$ steps), 
is at least $\nu^{d} \cdot (1-\epsilon)$.
We have:
\begin{center}
$1- \epsilon \leq X^{\beta'}_{n+d}((\q,K')) \leq 1-\nu^{d} \cdot (1-\epsilon)$
\end{center}
This  gives $\epsilon \geq \frac{\nu^{d}}{1+\nu^{d}}$ which is a contradiction.
Therefore, we have constructed a recurrent cyclic set $G$ for the cycle, and have 
shown that one element of $G$ is singleton.

 \item We can see that $G$ is the only recurrent  cyclic set.  
Each state included in $S$ reaches, at most in \textit{MaxPath} steps,  
one state of $J$. Hence, the cell $S$, as the first element of $C^{\P}$, 
can not have another subset $g'$ constructing another  recurrent cyclic set.

\end{enumerate} 
We have proved that for a weakly synchronizing  MDP $M$, the perfect information 
subset construction for $M$, has a cycle $C^{\P}$ such
that $\abs{\Delta(C^{\P})}=1$, and for $G \in \Delta(C^{\P})$, there exists 
$g \in G$ such that $\abs{g}=1$.
\begin{flushright}
$\Box$\end{flushright}
\end{proof}


\begin{example}
The MDP depicted in \figurename~\ref{f3} (the initial distribution
is $\mu_0(1)=1$ and $\mu_0(i)=0$ for $i \in \{2,\dots,5 \}$) with
strategy $\alpha$ which defined as follows $\alpha((L \times
\Sigma)^{*}L)(\sigma)=1/2$ for $\sigma \in \Sigma$,
 is weakly synchronizing. Note that $L=\{1,\dots ,9\}$, $\Sigma =\{
\sigma_1,\sigma_2 \}$.
\end{example}


\begin{figure}[t]
\begin{center}
\includegraphics[scale=0.40]{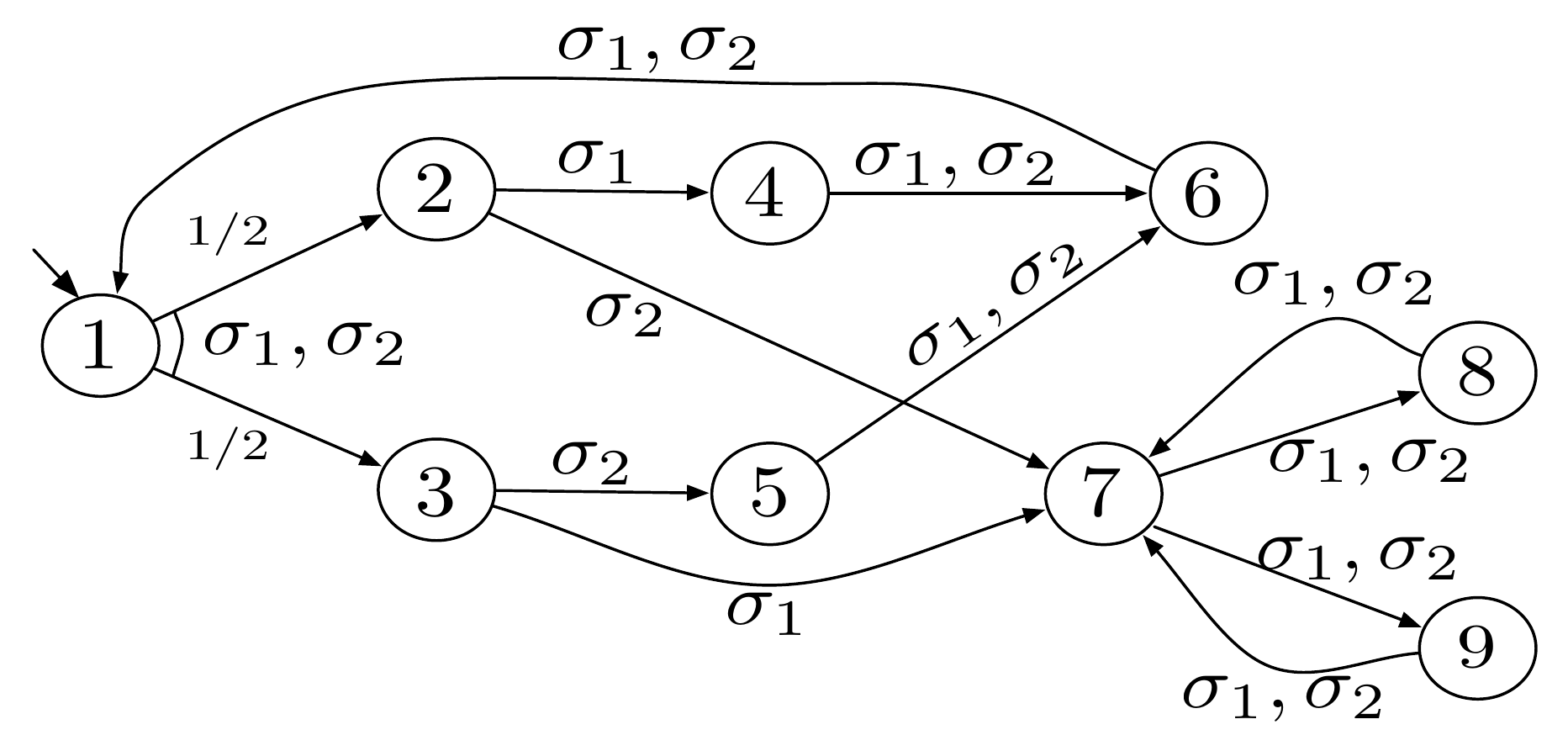}
\caption{A weakly
synchronizing MDP.}\label{f3}
\end{center}
\end{figure}


\section{Synchronizing objectives for Blind strategies}

We have defined a blind one-player stochastic game where the player is not
allowed to observe the current state of the game. We use a
characterization of synchronizing blind strategies to show that
the existence of synchronizing blind strategies can be decided. We
first present an example where the player is blind and has a
strategy to make the game synchronizing.


\begin{example}
The MDP depicted in \figurename~\ref{f2} (the initial distribution
is $\mu_0(1)=1$ and $\mu_0(i)=0$ for $i \in \{2,\dots,5 \}$) with
blind strategy $\alpha$ which defined as following
 $\alpha((L \times
\Sigma)^{*}L)(\sigma)=1/2$ for $\sigma \in \{\Sigma\}$
 is
strongly synchronizing. Note that $L=\{1,\dots ,8\}$, $\Sigma =\{
\sigma_1,\sigma_2 \}$.
\end{example}


\begin{figure}[b]
\begin{center}
\includegraphics[scale=0.40]{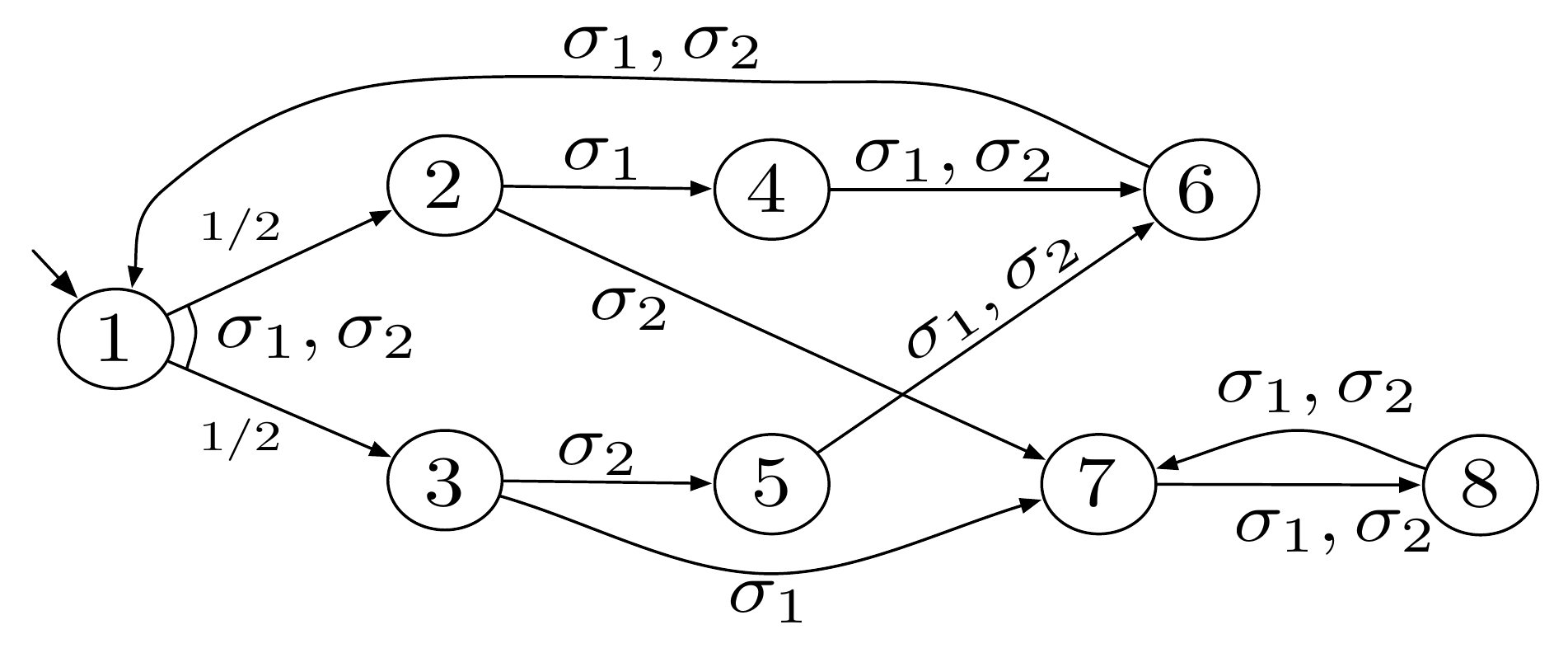}
\caption{An MDP that with some blind strategy is strongly
synchronizing.}\label{f2}
\end{center}
\end{figure}


\begin{theorem}\label{theo-blind-col-r1}
For a blind game over an MDP~$M$, there exists a
strategy~$\alpha$ such that $M$ with strategy~$\alpha$
 is \textbf{strongly synchronizing},
if and only if the blind subset construction
$M^{\B}$ for $M$, has an accessible  cycle $C^{\B}$ such
that $\abs{\Delta(C^{\B})}=1$, and for $G \in \Delta(C^{\B})$ and
for all $g \in G$, $\abs{g}=1$.
\end{theorem}


\begin{proof}
\noindent \emph{Sufficient condition.} We suppose that the blind
subset construction $M^{\B}=\tuple{\mathcal{L},L_{I}, \Sigma,
\delta^{\B}}$ for $M$,
has an accessible  cycle $C^{\B}$ such
that $\abs{\Delta(C^{\B})}=1$,
and for $G \in \Delta(C^{\B})$ and for all $g \in G$, $\abs{g}=1$.
Since this cycle is accessible,
then there exists a finite path
$P=p_0~\sigma'_0 ~\dots \sigma'_{m-2}~ p_{m-1}\sigma'_{m-1}~
p_{m}$ in $M^{\B}$ from $ p_0 = L_I$ to $ p_{m}=s_0$.
Consider the pure blind strategy $\alpha$ as follows
\begin{displaymath}
  \alpha(k) = \left\{
     \begin{array}{ll}
        \sigma'_k &  $if$~ 0 \leq k < m,\\
      \sigma_{(k-m)~mod~d} &  $if$~ m \leq k.
     \end{array}
   \right.
\end{displaymath}

Let us, construct a Markov chain $M'$ similar to which presented
in proof of Theorem \ref{theo-perf-col-r1}, with the below
probability function:
the probability transition function $\delta'$ is
defined as follows
\begin{displaymath}
  \delta'((i,\ell))((i',\ell'))= \left\{
     \begin{array}{ll}
     \delta(\ell,\sigma'_i)(\ell') &$if$~ (0 \leq i < m),~(i'=i+1),~ (\ell \in p_i) $ and $ (\ell' \in p_{i'}),\\
      \delta(\ell,\sigma_{i-m})(\ell') &$if$~ (m \leq i < m+d) ,~(i'=m+(i-m+1)~mod~d),\\
      &\hfill ~ (\ell \in s_{i-m}) $ and $ (\ell' \in s_{i'-m}),\\
0 & otherwise.
     \end{array}
   \right.
\end{displaymath}

\indent Suppose that $G \in \Delta(C^{\B})$ is the only recurrent cyclic set of the cycle, and it includes $d$ elements as $g_0, \dots g_{d-1}$.
Let $R$ be the set of states  $(m+i,\ell)$ such that $\ell \in g_i$, for $0 \leq i < d$.
As we have shown in proof of Theorem
\ref{theo-perf-col-r1}, the states of $R$ are the only recurrent states in the Markov chain $M'$.

Based on Proposition \ref{proposition-mc}, for the transient states $(k,\ell)$,
the probability $X_n((k,\ell))$ vanishes for $n \to \infty$.
Since for all $g \in G$, we have $\abs{g}=1$,
the support of $X_n$ ($n>m$) contains only one recurrent state.
Thus, the probability mass accumulates in that state:
for all $\epsilon>0$, for all $n>n_0$ there is a state $(i,\ell)$ with
$X_n((i,\ell))> 1-\epsilon$, that is $\norm{X^{\alpha}_n} >
1-\epsilon$.
Hence, $\lim_{n \to \infty} \norm{X^{\alpha}_n} = 1$
and~$M'$ is strongly synchronizing.
Therefore, so is the MDP $M$ under the blind strategy $\alpha$.


\noindent \emph{Necessary condition.} We benefit from arguments presented in Proof of Theorem \ref{theo-perf-col-r1}; but here since the winning strategy is blind, we use blind subset constructions. 
\begin{flushright}
$\Box$\end{flushright}
\end{proof}


\begin{theorem}\label{theo-blind-col-r2}
For a blind game over an MDP~$M$, there exists a
strategy~$\alpha$ such that $M$ with strategy~$\alpha$
 is \textbf{weakly synchronizing},
if and only if the blind subset construction
$M^{\B}$ for $M$, has an accessible  cycle $C^{\B}$ such
that $\abs{\Delta(C^{\B})}=1$, and for $G \in \Delta(C^{\B})$,
there exists $g \in G$ such that~$\abs{g}=1$.
\end{theorem}

\begin{proof}
\noindent \emph{Sufficient condition.} We suppose that the blind
subset construction $M^{\B}=\tuple{\mathcal{L},L_{I}, \Sigma,
\delta^{\B}}$ for $M$, has an accessible  cycle $C^{\B}$ such
that $\abs{\Delta(C^{\B})}=1$, and for $G \in \Delta(C^{\B})$,
there exists $g \in G$ such that~$\abs{g}=1$.

Consider a pure strategy similarly to which presented in proof of Theorem
\ref{theo-perf-col-r1}. Let us, here as well, construct the Markov
chain $M'$, and therefore discuss on transient and recurrent
states of $M'$.

Suppose that $G \in \Delta(C^{\B})$ is the only recurrent cyclic set of the cycle, and it includes $d$ elements as $g_0, \dots g_{d-1}$.
Let $R$ be the set of states  $(m+i,\ell)$ such that $\ell \in g_i$, for $0 \leq i < d$.
As we have shown in proof of Theorem
\ref{theo-perf-col-r1}, the states of $R$ are the only recurrent states in the Markov chain $M'$.
Suppose that $p$ is the probability to be in one
state of $R$ at step $n$.
 Based on Proposition \ref{proposition-mc}, for the transient
states $(i,\ell)$, the probability $X_n((i,\ell))$ vanishes for $n
\to \infty$, which leads $lim_{n\rightarrow \infty}\,p=1$. On the
other hand, by hypothesis, for $G \in \Delta(C^{\B})$, there
exists $g_j \in G$ ($ 0 \leq j <d$) such that $\abs{g_j}=1$. Then
every $d$ steps, at least once, the whole of probability $p$
gathers in only one state $(m+j,\ell)$ where $ \ell \in g_{j}$. As a result, for all
$k \in \mathbb{N}$,
$\max(\norm{X_{m+d.k}^{\alpha}},\norm{X_{m+d.k+1}^{\alpha}},...\norm{X_{m+d.k+d-1}^{\alpha}})
> p$. We have shown that $lim_{n\to \infty} \,p=1$, hence
$limSup_{n \to \infty} \, \norm{X_{n}^{\alpha}}=1$.

\noindent \emph{Necessary condition.} We benefit from arguments presented in Proof of Theorem \ref{theo-perf-col-r2}; but here since the winning strategy is blind, we use blind subset constructions. 
\begin{flushright}
$\Box$\end{flushright}
\end{proof}

From the four previous theorems, we obtain the following result.
\begin{theorem}\label{theo:decidability}
The problem of deciding the existence of a
$\{$perfect-information, blind$\}$ strategy in MDPs for a
$\{$strongly, weakly$\}$ synchronizing objective is decidable.
\end{theorem}

We have defined a new class of objectives for Markov decision
processes, and we have given a decidable characterization of
winning strategies for these objectives. Further investigations
will be devoted to studying the precise complexity of the problem,
establishing memory bounds, and extending this framework to
partially-observable MDPs and stochastic two-player games.


\bibliographystyle{eptcs}
\bibliography{biblio}
\end{document}